\documentclass[letterpaper, 10pt, conference]{ieeeconf} 
\IEEEoverridecommandlockouts
\usepackage{cite}
\usepackage{url}
\usepackage{textcomp}

\usepackage{amsmath, amssymb, amssymb, amsthm}

\usepackage{algorithm2e}
\usepackage[noend]{algpseudocode}
\usepackage{tabulary}
\usepackage{color}
\usepackage{tikz, pgfplots}
\usepackage{tabularx}

\usepackage{mathtools}
\mathtoolsset{showonlyrefs}

\newcommand{\ul}{\underline}

\newtheorem{prop}{Proposition}
\newtheorem{theorem}{Theorem}

\usepackage{comment}

\usetikzlibrary{calc,shapes.callouts,shapes.arrows}
\usepgfplotslibrary{fillbetween}
\pgfplotsset{compat=newest}
\usetikzlibrary{calc}
\usetikzlibrary{positioning}

\usetikzlibrary{external}

\definecolor{mycolor1}{RGB}{230,97,1}
\definecolor{mycolor2}{RGB}{178,171,210}
\definecolor{mycolor3}{RGB}{253,184,99}
\definecolor{mycolor4}{RGB}{94,60,153}

\DeclareMathOperator*{\argmin}{argmin}
\newcommand{\ol}{\overline}

\usepackage{etoolbox}
\newtoggle{extended}
\toggletrue{extended}

\begin{document}

\title{\LARGE \bf On the Impact of the Capacity Drop Phenomenon for \\ Freeway Traffic Flow Control}

\author{Michael E. Cao, Gustav Nilsson, and Samuel Coogan
\thanks{The authors are with the School of Electrical and Computer Engineering, Georgia Institute of Technology, Atlanta, 30332, USA {\tt\small \{mcao34, gustav.nilsson, sam.coogan\}@gatech.edu}. S. Coogan is also with the School of Civil and Environmental Engineering, Georgia Institute of Technology, Atlanta, 30332, USA.}
\thanks{
This research was supported in part by the National Science Foundation under award \#1749357.}
}

\maketitle

\begin{abstract}
Capacity drop is an empirically observed phenomenon in vehicular traffic flow on freeways whereby, after a critical density is reached, a state of congestion sets in, but the freeway does not become decongested again until the density drops well below the critical density. This introduces a hysteresis effect so that it is easier to enter the congested state than to leave it. However, many existing first-order models of traffic flow, particularly those used for control design, ignore capacity drop, leading to suboptimal controllers. In this paper, we consider a cell transmission model of traffic flow that incorporates capacity drop to study the problem of optimal freeway ramp metering. We show that, if capacity drop is ignored in the control design, then the resulting controller, obtained via a convex program, may be significantly suboptimal. We then propose an alternative model predictive controller that accounts for capacity drop via a mixed integer linear program and show that, for sufficiently large rollout horizon, this controller is optimal. We also compare these approaches to a heuristic hand-crafted controller that is viewed as a modification of an integral feedback controller to account for capacity drop. This heuristic controller outperforms the controller that ignores capacity drop but underperforms compared to the proposed alternative model predictive controller. These results suggest that it is generally important to include capacity drop in the controller design process, and we demonstrate this insight on several case studies. 

\end{abstract}

\section{Introduction}
\label{section:intro}

Many transportation networks today are utilized close to or over their capacity. Different types of actuators such as ramp metering and variable speed limits are therefore used to avoid congestion effects, since those congestion effects will lead to an under-utilization of the transportation network's capacity.

This loss of capacity is due to the fact that when the traffic density increases after a certain point, the traffic flow decreases. This phenomena is empirically observed, and incorporated in many of the models for traffic, such as the LWR-model~\cite{lighthill1955kinematic,richards1956shock} and the cell transmission model (CTM)~\cite{daganzo94ctm}. 

While the classical LWR and CTM models assume that the decrease in flow capacity happens gradually with an increase of traffic density, recent work suggests that this drop in flow capacity happens abruptly~\cite{zhang2010, srivastava2013, yuan2015}, and several different ways to model this phenomena have been suggested~\cite{KONTORINAKI201752}. This abrupt drop in flow capacity is commonly referred to as a \emph{capacity drop}. Usually, there is a hysteresis effect associated with the capacity drop, as suggested in, e.g.,~\cite{MONAMY20121388, torne2014coordinated}. When the capacity drop occurs, the traffic saturation enters a congested state, and to leave the congested state---a state with a significantly lower throughput---the traffic density needs to be driven to a lower level than before the capacity drop. The consequence of this effect is that it is much easier to steer the system into the congested state, e.g., through non-robust ramp metering, than to steer the system out of it. 

Though they can be interpreted as separate phenomena, capacity drop and hysteresis are linked so closely that it is assumed the existence of either implies the existence of the other. Thus, this paper uses the terms capacity drop and hysteresis to refer to the same overall behavior. For a full hybrid dynamic description of how to model the capacity drop, we refer the reader to~\cite{burden2018contraction}. 

The focus of this paper is to qualify the importance of taking the hysteresis effect into account when designing ramp metering controllers. While previous simulation studies in~\cite{maggi2015} suggest that it may not be necessary for the controller to take the capacity drop phenomenon into account, our theoretical results quantify when taking the hysteresis effect into account will improve the performance. Under such settings, our results show that incorporating the capacity drop phenomenon in the controller design yields a considerable performance gain during high-demand scenarios. It is worth noting that the need for appropriate control action is the highest during those scenarios.

To illustrate this effect, we implement a model predictive controller (MPC) that takes the capacity drop into account. While MPCs are common for both ramp metering and variable speed limit control on highways~\cite{hegyi2005mpc, gomes2016freeway, muralidharan2012optimal}, and the control actions can be computed efficiently through convex relaxations~\cite{como2016convexity}, these controllers are often developed for non-hysteretic dynamics, and hence they completely neglect capacity drop. In this paper, we show that naively applying an MPC without modelling the hysteretic effect can, depending on the magnitude of the capacity drop, lead to a significant loss of throughput on the freeway. \iftoggle{extended}{}{Due to page constraints, we omit full proofs of the propositions, and refer the reader to the extended version of this paper\footnote{\url{https://arxiv.org/abs/21XX.0XXXX}}.}

The outline of the paper is as follows: In the next section, Section~\ref{section:model formulation}, we introduce basic notation and present the model for traffic flow on a highway with capacity drop. In Section~\ref{section:controllers}, we propose several controllers for ramp-metering, two of which account for capacity drop, and one that ignores it, to be used for comparison. In Section~\ref{section:theoretical_results} we make formal statements about how much performance one can lose by not taking the capacity drop into account when synthesizing controllers. These results also answer the question of when it is important to incorporate the capacity drop in the controller. Additionally, we outline a method for calculating a sufficient horizon for a model predictive controller that accounts for capacity drop. In Section~\ref{section:case studies}, we show how each of the proposed control strategies perform in several different case studies. The paper is concluded with suggestions for future work.

\section{Model and Problem Formulation}
\label{section:model formulation}
Freeway networks are often modelled with a first order model such as Daganzo's Cell Transmission Model (CTM)~\cite{daganzo94ctm}. In this section, we recall the standard CTM model for traffic flow on a freeway and then modify the model to account for the capacity drop phenomenon by introducing an additional binary state variable that indicates the congestion state of a cell.

\subsection{Cell Transmission Model}
We consider a length of freeway divided into $N$ cells, as exemplified in Fig.~\ref{fig:freeway}, numbered 1 to $N$ such that traffic flows from cell \(i\) to cell \(i+1\). In addition, cells may have an attached onramp from which cars can enter the cell from outside the system. There does not necessarily exist an onramp for every cell, and existing ones are numbered with the same index as the cell they are attached to (i.e. onramp 3 is attached to cell 3, regardless if it is the third onramp in the system).

\begin{figure}
    \centering
    \begin{tikzpicture}[scale=0.8]
    \draw[thick] (0,0) -- (7,0);
    \draw[thick] (0,-1) -- (1, -1) -- (0.5, -1.5);
    \draw[thick] (1.25, -1.5) -- (1.75, -1) -- (3, -1) -- (2.5, -1.5);
    \draw[->] (0.87, -1.5) -- ++(0.25,0.25);
    \begin{scope}[shift={(2,0)}]
    \draw[thick] (1.25, -1.5) -- (1.75, -1) -- (3, -1) -- (2.5, -1.5);
    \draw[->] (0.87, -1.5) -- ++(0.25,0.25);
    \end{scope}
    \begin{scope}[shift={(4,0)}]
    \draw[thick] (1.25, -1.5) -- (1.75, -1) -- (3, -1);
    \draw[->] (0.87, -1.5) -- ++(0.25,0.25);
    \end{scope}
    \draw[dashed] (0, -0.5) -- (7, -0.5);
    \draw[->] (0.0, -0.25) -- ++(0.25,0);
    \draw[->] (0.0, -0.75) -- ++(0.25,0);

    \draw[dotted] (3, 0.5) -- (3, -1); 
    \draw[dotted] (5, 0.5) -- (5, -1);
    
    \draw[|->] (0, 0.5) -- node[above] {$\text{cell }1$} (3, 0.5);
    \draw[|->] (3, 0.5) -- node[above] {$\text{cell }2$} (5, 0.5);
    \draw[|->] (5, 0.5) -- node[above] {$\text{cell }3$} (7, 0.5);

    \end{tikzpicture}
    \caption{Example of a freeway segment with three cells.}
    \label{fig:freeway}
\end{figure}
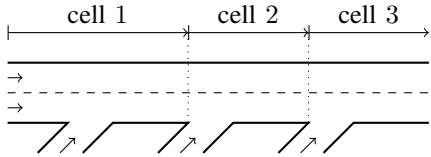

The density of traffic in each cell is denoted by \(x_i \geq 0\) and the density of traffic in each onramp is denoted by \(r_i \geq 0\). Cells each have an associated \textit{supply function} 
\begin{equation}
\label{eq:supply}
s_i(x_i) = -w_i(x_i - x^{\text{jam}}_i) \,,
\end{equation}
where \(w_i > 0\) denotes the shock-wave speed and \(x^{\text{jam}}_i > 0\) denotes the jam density of the cell, which gives the maximum possible inflow to the cell when the density of cars in the cell is equal to \(x_i\). Cells also each have an associated \textit{demand function} 
\begin{equation}
\label{eq:demand}
    d_i(x_i)= v_ix_i\,,
\end{equation}
where \(v_i > 0\) denotes the free-flow speed of the cell, which gives the maximum possible outflow from the cell when the density of cars in the cell is equal to \(x_i\).

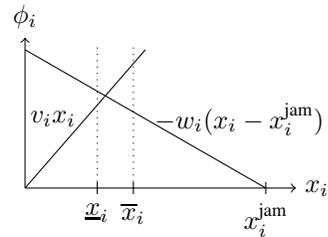
\begin{figure}
    \centering
    
    \begin{tikzpicture}[scale=0.8]
    \draw[->] (0,0) -- (4.5,0) node[right] {$x_i$};
    \draw[->] (0,0) -- (0,2.5) node[above] {$\phi_i$};
    \draw (0,0) -- node[left] {$v_ix_i$} (2,2.3);
    \draw (0,2.3) -- node[right] {$-w_i(x_i - x_i^{\text{jam}})$} (4,0);
    \draw (4,0.1) -- (4,-0.1) node[below] {$x_i^{\text{jam}}$};
    \draw (1.2,0.1) -- (1.2,-0.1) node[below] {$\ul{x}_i$};
    \draw (1.8,0.1) -- (1.8,-0.1) node[below] {$\overline{x}_i$};
    \draw[dotted] (1.2,0) -- (1.2, 2.4);
    \draw[dotted] (1.8,0) -- (1.8, 2.4);

    \end{tikzpicture}
    
    \caption{The supply and demand functions for a cell $i$. The outflow from the cell is limited by the demand function while the inflow to the cell is limited by the supply function.}
    \label{fig:fundamentaldiagram}
\end{figure}

Thus far, the formulation has matched that of a standard CTM. To incorporate capacity drop into the system,  we must now introduce new parameters into the model. Each cell has two associated densities, \(\overline{x}_i\) and \(\underline{x}_i\), such that  \(\overline{x}_i\geq\underline{x}_i>0\) which are the densities at which the cell becomes congested and decongested, respectively. Fig.~\ref{fig:fundamentaldiagram} summarizes the parameters associated with each cell. While cell \(i\) is decongested, it is able to accept any amount of inflow from cell \(i-1\). When cell \(i\) becomes congested, the inflow capacity of cell \(i\) becomes limited by its supply, thus limiting the respective outflow of cell \(i-1\). We denote the congestion state of each cell using the binary variable \(\sigma_i\), which is defined formally later in this section. This is the core addition to the standard CTM that allows us to model capacity drop in the system.

Let \(\phi_i > 0\) denote the outflow from a cell; $\phi_i$ will be defined in terms of previously defined quantities subsequently. To capture offramps, each cell has an associated fixed \(\beta_i \in [0, 1]\) which denotes the ratio of cars that progress onto the next cell as opposed to leaving the system through an offramp. Therefore \(\beta_i = 1\) is equivalent to a cell having no offramp. The inflow of cars into a cell \(i\) from the previous cell \(i-1\) is then \(\beta_{i-1}\phi_{i-1}\). Additionally, cell 1 can receive upstream inflow from outside the system, which we denote as $\lambda_0$.

Onramps do not have the same variables and functions associated to them as cells. Instead, onramps are assigned a \textit{maximum outflow capacity} \(c_i\) representing the maximum possible flow of cars from the onramp to its associated cell, and a \textit{control signal} \(u_i\in[0, 1]\) representing the controlled rate of cars that enter the onramp's associated cell such that the actual rate of cars that enter is \(u_i c_i\) assuming enough cars exist on the onramp. If there is no onramp attached to cell \(i\), the corresponding \(c_i\) is simply set to \(0\). Each onramp also has an associated inflow \(\lambda_i\) which denotes the flow of cars into the onramp from outside the system.

We assume the system is sampled over \(K\) total timesteps with sampling time interval \(h\), where \(h\) is set such that the CFL condition \cite{cfl2013} of the system is fulfilled, i.e. the sample time is not long enough to allow cars to traverse more than one cell during a timestep. Thus, to reference the density \(x_i\) or \(r_i\) of a cell or onramp at a specific timestep \(k\in[0,K]\), we utilize the notation \(x_i(k)\) or \(r_i(k)\). Likewise, we use \(\sigma_i(k)\) to denote a cell's congestion state at timestep \(k\), \(\phi_i(k)\) to denote outflow from a cell at timestep \(k\), \(\lambda_i(k)\) to denote the inflow of cars into the system at timestep \(k\), and \(u_i(k)\) to denote the control signal sent to an onramp at timestep \(k\).


\subsection{Modelling Capacity Drop}
In this subsection, we formally define the factors that allow us to model the capacity drop phenomenon in the system through a hysteresis effect as suggested in~\cite{burden2018contraction}. The previously mentioned congestion state $\sigma_i$ of a cell at timestep~\(k\) follows the dynamics
\begin{equation}
\label{eq:congestion}
\sigma_i(k) = \begin{cases}
            1 & \text{if } x_i(k) \geq \overline{x}_i\,,\\
            0 & \text{if } x_i(k) \leq \underline{x}_i\,,\\
            \sigma_i(k-1) & \text{otherwise}.
        \end{cases}
\end{equation}
If \(x_i(0) > \underline{x}_i\) and \(x_i(0) < \overline{x}_i\), \(\sigma_i(0)\) is assumed to be 0.

The outflow \(\phi_i(k)\) of cell \(i\) at timestep \(k\) is calculated depending on the congestion state of the downstream cell and given by
\begin{equation}
\label{eq:flow}
    \phi_i(k) = \begin{cases}
    \min\{d_i(x_i(k)), {\beta_i}^{-1}{s_{i+1}(x_{i+1}(k))}\}\\
    &\hspace{-60pt} \text{ if } \sigma_{i+1}(k) = 1 \,, \\
    d_i(x_i(k)) &\hspace{-60pt} \text{ if } \sigma_{i+1}(k) = 0\,.
    \end{cases}
\end{equation}
It should be noted that when the downstream cell is in a congested state, the supply function will always limit the outflow of the cell.

The new density of each cell at time \(k+1\) is calculated as
\begin{equation}
\label{eq:density_update}
\begin{split}
x_i(k+1) = x_i(k) &+ \Delta x_i(k) \,, \\
\Delta x_i(k) = h(\beta_{i-1}\phi_{i - 1}(k) - &\phi_i(k)) + \min[r_i(k), u_i(k) c_i] \,,
\end{split}
\end{equation}
and the new density of each onramp at time \(k+1\) is calculated as 
\begin{equation}\label{eq:onramp}
   r_i(k+1) = r_i(k) + \lambda_i(k) - \min[r_i(k), u_i(k) c_i] \,. 
\end{equation}

As cell 1 can receive cars from outside the system, the update equation $\Delta x_1(k) = h(\beta_{i-1}\phi_{i - 1}(k) - \phi_i(k)) + \min[r_i(k), u_i(k) c_i] + \lambda_0(k)$.

Thus, with the system fully defined, the problem that must be solved is the formulation of a controller that optimizes throughput of the system over \(K\) timesteps. More formally, we define the problem as the formulation of a controller that maximizes the objective function $\sum_{k=0}^{K-1}{\sum_{i=1}^{N}{h(1-\beta_i)\phi_i(k)}} $.

For modeling purposes, we note that this is equivalent to minimizing the number of cars present within the system over all time, or more formally, minimizing the objective function
\begin{equation}
    \label{eq:model_obj}
    \sum_{k=0}^{K-1}{\sum_{i=1}^{N}{x_i(k) + r_i(k)}}\,.
\end{equation}

In later sections we will show the effect that accounting for capacity drop, which we have integrated into this system, can have on performance. The proposed controllers are assessed as model predictive controllers so that they are formulated over a $\overline{T}$-step horizon, the first control input of the rollout is applied, then the next control input is computed via a new $\overline{T}$-step horizon rollout.

\section{Freeway Controller Formulation}
\label{section:controllers}
With model in hand, we now discuss three possible control strategies for the freeway network\footnote{
For the full implementations of each controller, we refer the reader to the following GitHub repository: \url{https://github.com/gtfactslab/Cao_CCTA2021}}. First, we introduce a controller that makes the assumption to ignore capacity drop in favor of producing a convex relaxation of the problem \cite{schmittthesis}. We then formulate a controller that accounts for capacity drop by implementing the system dynamics directly into a mixed integer program, of which there are many different solvers available. Finally, we introduce a controller that is hand-crafted to account for capacity drop to serve as a point of comparison.

\subsection{Relaxed Approximate MPC}
One control solution for ramp metering on a freeway is to ignore the capacity drop phenomenon and assume that outflow is always limited by both supply and demand \cite{schmittthesis}. This allows for a convex relaxation of the problem to be formulated and used in a model predictive controller~\cite{como2016convexity}. To illustrate this behavior, we mirror the steps taken by~\cite{schmittthesis} to create a convex relaxation given as
\begin{align*}
 \underset{x(k),\phi(k),u(k)}{\text{minimize}}\quad &\sum_{k=0}^{T-1}\sum_{i=1}^N x_i(k) + r_i(k)\\
 \text{subject to}\quad & \eqref{eq:supply},\eqref{eq:demand},\eqref{eq:density_update},\eqref{eq:onramp} &\forall i\in[1, N]\\
 & \phi_i(k) \leq d_i(x_i(k))\ &\forall i\in[1, N]\\
 & \beta_i\phi_i(k) \leq s_{i+1}(x_{i+1}(k))\ &\forall i \in [1,N-1]\\
 & \phi_i(k) \geq 0 \ &\forall i\in[1, N]\\
 & x_i(0),\ r_i(0)\ \text{given}\ &\forall i\in[1, N] \,.
\end{align*}
We thus refer to this formulation as the Relaxed Approximate MPC (RAMPC). Thanks to its convexity, this controller can calculate a quick solution using most modern solvers.

\subsection{Exact Hysteretic MPC}
Given the model formulation from Section~\ref{section:model formulation}, we introduce a controller that directly implements this system behavior, including capacity drop, into its calculations. Thus, the control problem is formulated as
\begin{align*}
 \underset{x(k),\phi(k),u(k)}{\text{minimize}}\quad &\sum_{k=0}^{T-1}\sum_{i=1}^N x_i(k) + r_i(k)\\
 \text{subject to}\quad & \eqref{eq:supply},\eqref{eq:demand},\eqref{eq:congestion},\eqref{eq:flow},\eqref{eq:density_update}, \eqref{eq:onramp}\ &\forall i\in[1, N] \,.
\end{align*}

This captures the system model exactly as formulated, modeling capacity drop, and thus 
we refer to this controller as the Exact Hysteretic MPC (EHMPC).

This results in a mixed integer linear program, which can be implemented in several modern solvers. In this instance, it was implemented in the Gurobi solver \cite{gurobi}. 

\subsection{A Heuristic Controller}
Given the complexity of the Exact Hysteretic MPC, it is natural to ask if a hand-crafted controller constructed using domain-specific knowledge is capable of accommodating the capacity drop phenomenon and achieving optimal or near-optimal performance with less computational overhead than the Exact Hysteretic MPC. Here, we describe such a controller, referred to as the Heuristic Controller (HC). While the full formulation cannot be included due to page limitations, a general overview of the controller's behavior is provided.

The Heuristic Controller uses the parameters of the system to calculate what it believes to be the optimal densities for each cell. At each timestep, it calculates what the state of the system will be during the next timestep, then adds or withholds cars from their respective onramps to drive the attached cells to their calculated optimal density. This controller is suboptimal, as the calculations are manually tuned, but the controller is able to account for capacity drop by factoring in the congestion state of each cell into its optimal density calculations (i.e. if a cell is congested and the optimal densities require that it not be, it will withhold cars from each onramp attached to a cell that comes before it in an attempt to decongest the cell).

This controller is not considered a major contribution in itself, hence the omission of the full formulation. Rather, it is provided to serve as another point of comparison to the naive approach and illustrate that suboptimally accounting for capacity drop can still result in performance gains.

\medskip

With the controllers formulated, we can now discuss the theoretical gains to be had in accounting for capacity drop as well as some considerations to be made when running the controllers. We will see that ignoring capacity drop induces suboptimality when the control actions are applied to the actual system. Later sections will test each controller in simulations and compare their performance.

\section{Theoretical Results}
\label{section:theoretical_results}
Given the three control strategies defined above, we now consider 
the potential gain in implementing capacity drop and conditions for which accounting for capacity drop is beneficial. We also derive a sufficient horizon length for a hysteretic controller to produce the optimal control action. In addition, we discuss
the generalizability of the calculation to different supply and demand functions. Due to the complex nature of accounting for capacity drop, we discuss these findings for a two-cell system \(i\in[1, 2]\), which is the simplest system that can exhibit capacity drop. We then consider larger networks in the case studies below.

\subsection{Potential Gain in a Two-Cell System}
Consider a two-cell freeway network, $N=2$. We start by defining a condition for which moving cell $2$ from a congested to a decongested state would increase throughput.

We will in this section restrict ourselves to study the specific case when the exogenous inflows are, for a period of time, more than what can potentially flow out from the highway, i.e., cases where more cars flow into the onramps than can exit them every time step, $\lambda_i \geq c_i$ for all $i$. While this situation would cause infinite back spill on the onramps in the long run, we want to be able to fully compare the performance of both controllers without being affected by system limitations.
\begin{prop}\label{prop:decongest}
Consider a freeway system with dynamics~\eqref{eq:supply}--\eqref{eq:density_update} with $N = 2$ and  $\lambda_2 =0, r_2(0) = 0$, 
and suppose that $\lambda_0$ is such that
\begin{equation} \label{eq:fill_condition}
\lambda_0 + c_1 > hd_1(x_1^*) \,,
\end{equation}
where $x_1^* = v_2\overline{x}_2/(\beta_1v_1)$ and 
\begin{equation} \label{eq:drain_condition}
\beta_1\lambda_0 < hd_2(\ul{x}_2) \,.
\end{equation}
Let
\begin{equation}
x_c = \frac{w_2x^{\text{jam}}_2}{v_2+w_2} \,.     \label{eq:x_c}
\end{equation}
If $x_c < \ol{x}_2$, then provided that the rollout horizon is sufficiently long, the Exact Hysteretic MPC enables higher throughput than the Relaxed Approximate MPC.

\end{prop}

%


    
\iftoggle{extended}{
\begin{proof}
For the two cell system, in order to achieve an optimal steady-state output, a controller must drive the density of the cells to a point where the inflow into cell 2 matches its outflow. For the Relaxed Approximate MPC, this requirement takes the form of
\begin{equation}
    \label{eq:convex_steadystate}
    \beta_1\min\{d_1(x_1), s_2(x_2)/\beta_1\} = d_2(x_2) \, .
\end{equation}

The maximum value attainable for $x_2$ that fulfills this requirement is such that $s_2(x_2) = d_2(x_2)$. It can be verified that $x_2 = x_c$, where $x_c$ is given in~\eqref{eq:x_c}, satisfies this equality. Moreover, since~\eqref{eq:fill_condition} holds, $x_1$ will be large enough such that $\min\{d_1(x_1), s_2(x_2)/\beta_1\} = s_2(x_2)/\beta_1$. This because the inflow to cell $1$ is enough to reach this point.



Thus, we know that the maximum possible steady-state outflows are
\begin{equation}
    \label{eq:convex_flows}
    \begin{split}
   \phi_1 &= s_2(x_c)/\beta_1 = d_2(x_c)/\beta_1 = v_2x_c/\beta_1 \,, \\
   \phi_2 &= d_2(x_c) = v_2x_c\,, 
   \end{split}
\end{equation}
and maximum possible number of cars exiting the system per timestep is
\begin{equation}
\label{eq:convex_cars_leaving}
         E_{\text{convex}} = h((1 - \beta_1)\phi_1 + \phi_2) 
        = h\frac{v_2x_c}{\beta_1}\,.\\
\end{equation}


By taking advantage of the hysteresis effect, the optimal steady-state requirement becomes $\beta_1d_1(x_1) = d_2(x_2)$, when $x_2 \leq \ol{x}_2$.
Throughput is maximized when the density of cell $2$  is \(\overline{x}_2\), as this is the point where the flow between cells becomes limited by the supply of cell 2, past which the requirement reverts back to the form outlined in \eqref{eq:convex_steadystate}.
Thus, in order to supply cell $2$ with enough cars, the density of cell 1 should be such that
\begin{equation}
    \label{eq:x1_target}
   x_1 = \frac{v_2}{\beta_1v_1}\overline{x}_2\,,
\end{equation}
which results in outflows of $\phi_1 = v_1\frac{v_2}{\beta_1v_1}\overline{x}_2 = \frac{v_2}{\beta_1}\overline{x}_2$ and $\phi_2 = v_2\overline{x}_2$.

The number of cars leaving this system at each timestep can thus be calculated as
\begin{equation}
    \label{eq:hyst_cars_leaving}
         E_{\text{hyst}} = h((1 - \beta_1)\phi_1 + \phi_2)\\
        = h\frac{v_2\overline{x}_2}{\beta_1} \, .
\end{equation}

The increase in throughput gained by accounting for hysteresis can then be quantified as
\begin{equation}
    \label{eq:throughput_gain}
         \Delta E = E_{\text{hyst}} - E_{\text{convex}}
        = h\frac{v_2}{\beta_1}(\overline{x}_2-x_c)\,.
\end{equation}

Thus, so long as $\overline{x}_2 > x_c$, and the hysteretic controller has enough control over the system to drive each cell to the desired density, the hysteretic controller will strictly outperform the convex relaxed controller. 


For the system to be able to decongest, it must be possible for $x_2 < \underline{x}_2$, thus the minimum inflow into the system must be lower than the outflow experienced by cell 2 when decongesting. The condition~\eqref{eq:drain_condition} guarantees this.
%
\end{proof}


}{}

Proposition~\ref{prop:decongest} states that in certain cases, the hysteretic MPC will yield a better performance if the rollout horizon is large enough. In the next section, we discuss a method for deriving a sufficient rollout horizon to guarantee a performance increase.

\subsection{Sufficient Horizon Requirements}
\begin{figure}
\centering
\begin{tikzpicture}[scale=0.7]
\draw[->] (0,0) -- node[below] {Time} (7,0);
\draw[->]  (0,0) --  (0,4);
\node[rotate=90] (o) at (-0.3,2)  {Outflow};
\coordinate (a0) at (0,3.5);
\coordinate (a1) at (1.25,1.5);
\coordinate (a2p) at (5,1.5);
\coordinate (a2) at (7,1.5);
\coordinate (b0) at (0,3.5);
\coordinate (b1) at (1.5,0.5);
\coordinate (b1p) at (2.25, 1.5);
\coordinate (b2) at (3,2);
\coordinate (b2p) at (5,2);
\coordinate (b3) at (7,2);

\draw [name path=MPC, dashed] (a0) to[out=-70,in=180] (a1) to[out=0,in=180] (b1p);
\draw [name path=MPC2, dashed] (b1p) to (a2p);
\draw [dashed] (a2p) to (a2);

\draw [name path=HMPC] (b0) to[out=-70,in=150] (b1) to[out=40,in=-125] (b1p);
\draw [name path=HMPC2] (b1p) to[in=-180] (b2) to[out=0,in=180] (b2p);
\draw [name path=HMPC3] (b2p) to (b3);

\draw[dotted, name path=A] (1.5,0) -- (1.5, 3.6);
\draw[dotted] (3,0) -- (3, 3.6);
\draw[dotted] (5,0) -- (5, 3.6);

\draw[<->] (0, 3.6) -- node[above] {(a)} (1.5, 3.6) ;
\draw[<->] (1.5, 3.6)  -- node[above] {(b)} (3, 3.6); 
\draw[<->] (3, 3.6) -- node[above] {(c)} (5, 3.6);

\tikzfillbetween[of=MPC and HMPC]{mycolor1, opacity=0.2};
\tikzfillbetween[of=MPC2 and HMPC2]{mycolor4, opacity=0.2};
\end{tikzpicture}
\caption{Example of how the outflow evolves with time for a two cell system controlled by both the Relaxed Approximate MPC (dashed) and Exact Hysteretic MPC (solid). The Exact Hysteretic MPC needs to go through three different stages to improve its objective of increased throughput. (a) denotes Decongestion, (b) denotes Recovery, and (c) denotes Steady-State.}
\label{fig:behaviors}
\end{figure}
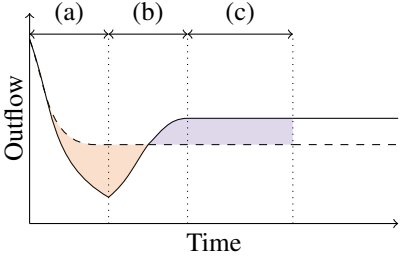
\label{section:horizon_req}
If the system fulfills the conditions set by Proposition~\ref{prop:decongest}, i.e.,~\eqref{eq:fill_condition} and \eqref{eq:drain_condition} are fulfilled and $\ol{x}_2 > x_c$, the optimal control strategy will decongest the system before returning to a state of higher throughput. As this temporarily results in a lower throughput, the rollout horizon for a hysteretic controller must be large enough to capture the long-term benefit of decongestion. The different stages of this process are shown in Fig.~\ref{fig:behaviors}. Thus, in this subsection we \iftoggle{extended}{first }{}present a theorem for deriving a sufficient rollout horizon needed to cover the worst-case scenario, assuming an average upstream inflow $\overline{\lambda}_0$ . \iftoggle{extended}{The rest of the subsection is dedicated to proving this theorem.}{}

\begin{theorem}
\label{theo:sufficient_horizon}
If the conditions for Proposition~\ref{prop:decongest} hold  then  a rollout horizon of $\overline{T}_D + \overline{T}_R + \overline{T}_S$ or greater, where
\begin{align}
    \label{eq:Td}
    \overline{T}_D &= \left \lceil{\frac{\beta_1x_1(0) + (x_c - \underline{x}_2)}{hv_2\underline{x}_2  - \beta_1\overline{\lambda}_0}}\right \rceil,\\
    \nonumber\overline{T}_R &= \argmin_k \{\overline{x}_2 - \underline{x}_2\leq-hv_2\underline{x}_2 \\   
    \hspace{-2.7cm} \label{eq:Tr}&+\sum_{i=1}^{k} [h(\beta_1v_1x_1(\overline{T}_D+i) - v_2x_2(\overline{T}_D+i))]\},\\
   \nonumber \overline{T}_S &= \\
  \label{eq:Ts}  &\left \lceil{\frac{(\overline{T}_D + \overline{T}_R)E_{\text{convex}}-(\overline{T}_D \overline{E}_{\text{decon}} + \overline{T}_R \overline{E}_{\text{recov}})}{E_{\text{hyst}} - E_{\text{convex}}}}\right \rceil \,
\end{align}

where $\lceil \cdot \rceil$ denotes the ceiling operator, and where
\begin{align}
        x_1(n) &= x_1(n-1) + c_1 -hv_1x_1(n-1)   + \lambda_0(n-1),\\
\overline{E}_{\text{decon}} &= h\left((1-\beta_1)v_1\frac{x_1(0)}{2} + v_2\frac{x_c-\underline{x}_2}{2}\right),     \label{eq:decongest_avg_out}\\
   \overline{E}_{\text{recov}} &= h\left((1-\beta_1)v_1\frac{1}{2}\frac{v_2}{\beta_1v_1}\overline{x}_2 + v_2\frac{\overline{x}_2-\underline{x}_2}{2}\right),     \label{eq:recovery_avg_out}
\end{align}
is sufficient for the Exact Hysteretic MPC to achieve optimal throughput.
\end{theorem}

\iftoggle{extended}{
In order to prove this theorem, we refer to Fig.~\ref{fig:behaviors} and note that there are essentially three behaviors that the system can engage in: decongestion, recovery, and steady state. In the worst case scenario, the system must go through all three behaviors to achieve optimal outflow. During steady-state, the system has no need to engage in either of the other two behaviors. If the system is below the optimal steady-state, the system will require recovery before staying in steady-state. If the system starts congested, it will require decongestion and recovery before staying at steady-state. Thus, we can split the sufficient horizon calculation into three distinct calculations $\overline{T}_D$, $\overline{T}_R$, and $\overline{T}_S$, one for each behavior respectively, and solve for each separately. We also note that equations \eqref{eq:decongest_avg_out} and \eqref{eq:recovery_avg_out} do not depend on any $\lambda$, as they are measures of the outflow of the system at optimal capacity. So long as the onramps are able to keep the system at this capacity (which is true so long as 
\eqref{eq:fill_condition} holds), \eqref{eq:decongest_avg_out} and \eqref{eq:recovery_avg_out} are independent of these inputs.
\subsubsection{Decongestion}
\label{section:decongest_horizon}
We begin by calculating a bound on the number of timesteps \(\overline{T}_D\) needed for decongestion.
\begin{prop}
\label{prop:T_D}
Given a two cell system that satisfies \eqref{eq:fill_condition} and \eqref{eq:drain_condition} from Proposition~\ref{prop:decongest} and $\ol{x}_2 > x_c$, $\overline{T}_D$ timesteps are sufficient to capture the decongestion behavior of the system, i.e. $x_2(\overline{T}_D)<\ul{x}_2$, where $\overline{T}_D$ matches the form in \eqref{eq:Td}.
\end{prop}
\begin{proof}
First, we must prove the existence of an upper bound on the number of timesteps that are required to capture decongestion. We accomplish this by proving that given decongestion must take place, the optimal control action is to minimize the number of timesteps needed for the behavior to complete.

We calculate the outflow for each timestep during the decongestion process as $E_{\text{decon}} = h(1-\beta_1) \cdot \left(\min{\left[v_1x_1(k), \frac{-w_2(x_2(k)-x_2^{\text{jam}})}{\beta_1}\right]}+ v_2x_2(k)\right)
$.

We then define an error function \(\text{err}(k)\) as the difference between the total outflow of cars during each timestep in the decongestion period and the the maximum outflow of cars that does not result in downstream congestion $E_{\text{hyst}}$, i.e., $\text{err}(k) = E_{\text{hyst}} - E_{\text{decon}}$.

By definition, \(x_1(k) + x_2(k)\) is monotonically decreasing during the decongestion process, meaning that \(\text{err}(k)\) is monotonically increasing during this period. Given a set number of cars that enter the system over the entire horizon, and the control objective~\eqref{eq:model_obj}
the loss in outflow modeled by \(\text{err}(k)\) negatively affects the objective every timestep. To minimize this, we must solve
\begin{equation}
\min_{\overline{T}_D}\sum_{k=0}^{\overline{T}_D}{\text{err}(k)} \quad \text{such that} \quad x_2(\overline{T}_D) \leq \underline{x}_2 \,.
\end{equation}

With the knowledge that an upper bound exists, we can now calculate a sufficient one. We consider the number of cars that must leave the system from each cell in order to decongest. If the system starts with \(x_1(0) \geq \overline{x}_1, x_2(0) \geq \overline{x}_2,\) then assuming all onramp inflow is cut off from the system, the second cell will outflow up to $\beta_1x_1(0) + \beta_1\sum_{k = 0}^{\overline{T}_D}{\lambda_0(k)} + (x_2(0) - \underline{x}_2)$
cars at a rate of \(hv_2\underline{x}_2\) cars per timestep in the worst case scenario. If the actual outflow is higher, the system will decongest faster, and if the outflow is lower, the system is already decongested. With the assumed average upstream inflow $\overline{\lambda}_0$, our solution becomes the smallest $\overline{T}_D$ that satisfies $\overline{T}_Dhv_2\underline{x}_2\geq\beta_1x_1(0) + \beta_1\overline{T}_D\overline{\lambda}_0 + (x_2(0) - \underline{x}_2)$,
and we can thus calculate an upper bound \(\overline{T}_D\) on how many timesteps it will take to decongest cell 2 through $\overline{T}_D = \left \lceil{\frac{\beta_1x_1(0) + (x_2(0) - \underline{x}_2)}{hv_2\underline{x}_2 - \beta_1\overline{\lambda}_0}}\right \rceil$.

We can tighten this bound by noting that in a congested system, the maximum outflow is achieved at the point where \(x_2(k) = x_c\). Thus, any controller attempting to minimize the previously defined objective, even if it does not account for capacity drop, will attempt to drive the system to this point. We can consequently consider the initial state of cell 2 during decongestion to be \(x_2(0) = x_c\), as this is the actual point where we want the behavior of the controller to differ from one that does not account for capacity drop.
Thus, our bound matches the form stated in \eqref{eq:Td}.
\end{proof}

\subsubsection{Recovery}
We then look at the number of timesteps \(\overline{T}_R\) needed for the recovery period. 

\begin{prop}
\label{prop:T_R}
Given a two cell system that satisfies \eqref{eq:fill_condition} and \eqref{eq:drain_condition} from Proposition~\ref{prop:decongest} and $\ol{x}_2 > x_c$, $\overline{T}_R$ timesteps is sufficient to capture the recovery behavior of the system, i.e. $x_2(\ol{T}_D + \ol{T}_R)=\ol{x}_2$, where $\overline{T}_R$ is of the form outlined in~\eqref{eq:Tr}.
\end{prop}

\begin{proof}
As before, we must prove the existence of an upper bound on the number of timesteps that are required to capture the behavior. We mirror the logic taken in the previous section and define $E_{\text{recov}} =h((1-\beta_1)\min{[v_1x_1(k), \frac{-w_2(x_2(k)-x_2^{\text{jam}})}{\beta_1}]} + v_2x_2(k))$ and
$\text{err}(k) = E_{\text{hyst}} - E_{\text{recov}}
$, and see that
$\min_{\overline{T}_R}\sum_{k=0}^{\overline{T}_R}{\text{err}(k)}$
is solved by the lowest \(\overline{T}_R\) such that \(\text{err}(\overline{T}_R) \leq 0\), thus proving the existence of an upper bound. 

Given the existence of an upper bound, we may now calculate a sufficient one. In this case, we assume the cells start at the worst-case state left over from the decongestion period: \(x_1(\overline{T}_D) = 0, x_2(\overline{T}_D) = \underline{x}_2\). Since the system is now decongested, we know each update equation to be limited by demand only, and we can use their decongested forms to adjust the update equations for \(x_2, x_1\) as
\begin{equation}
    \begin{split}
    x_2(k+1) =\ &x_2(k) + h(\beta_1v_1x_1(k) - v_2x_2(k)) \,,\\
    x_1(k) =\ &x_1(k-1)\\
    &+ \min{[c_1, r_1(k-1)]} + \lambda_0(k-1)\\
    &- hv_1x_1(k-1)
    \end{split}
\end{equation}
for $k > 1$ and $x_2(1) =x_2(0) + h(\beta_1v_1x_1(0) - v_2x_2(0))$.

Ultimately, the goal is for \(x_2(\overline{T}_D + \overline{T}_R) \geq \overline{x}_2\).
In our worst-case assumption, \(x_1(\overline{T}_D)=0\), and thanks to condition \eqref{eq:fill_condition}, \(r_1(k) \geq c_1\) for all time. Thus our constraint becomes the lowest integer value of \(\overline{T}_R\) that satisfies the form outlined in~\eqref{eq:Tr}. Note that this is effectively running a forward simulation of the system with full input from the onramps and calculating how many timesteps it takes for the system to congest.
\end{proof}

To reach the desired densities, we assume \(c_1\) is high enough to increase cell 1 to the desired density calculated in \eqref{eq:x1_target}. This effectively means that $ c_1 \geq h\frac{v_2}{\beta_1}\overline{x}_2 - \lambda_0$. If the inequality is not true, it is not possible for the system to achieve the optimal steady-state densities. This does not mean that it is no longer worth decongestion, only that the arrived steady-state will have suboptimal outflow compared to the maximum possible.

\subsubsection{Steady State}
As shown in Fig.~\ref{fig:behaviors}, the overall output of the system during decongestion and recovery is lower than it would be given a convex relaxed controller. Thus, we consider the amount of timesteps \(\overline{T}_S\) in steady-state that are required for the benefit of decongestion to manifest fully within the hysteretic controller's rollout horizon.
\begin{prop}
\label{prop:T_S}
Given a two cell system that satisfies \eqref{eq:fill_condition} and~\eqref{eq:drain_condition} from Proposition~\ref{prop:decongest} and $\ol{x}_2 > x_c$, \(\overline{T}_S\) or more timesteps in steady-state, where $\overline{T}_S$ is of the form shown in \eqref{eq:Ts}, is sufficient for the overall output of the system over $\overline{T}_D+\overline{T}_R+\overline{T}_S$ timesteps to be greater than the output a convex relaxed controller produces during the same number of timesteps.
\end{prop}

\begin{proof}First, we estimate the average outflow from the congestion and recovery periods from each cell by using the worst-case bounds from the previous two periods. Since the supply and demand functions are in a linear form, the average outflow during the decongestion period can be estimated using expression \eqref{eq:decongest_avg_out} and the average outflow during the recovery period can be estimated using expression \eqref{eq:recovery_avg_out}. Note that the expression in \eqref{eq:recovery_avg_out} utilizes the calculated target density for cell 1 from \eqref{eq:x1_target}.

We have quantified the number of cars leaving the system during each timestep in the form of \(E_{\text{hyst}}, E_{\text{convex}}, \Delta E\) from equations \eqref{eq:convex_cars_leaving}, \eqref{eq:hyst_cars_leaving}, and \eqref{eq:throughput_gain}. This allows us to formulate an expression for \(\overline{T}_S\), such that $\overline{T}_SE_{\text{hyst}} + \overline{T}_D \overline{E}_{\text{decon}} + \overline{T}_R \overline{E}_{\text{recov}}
    \geq (\overline{T}_S + \overline{T}_D + \overline{T}_R)E_{\text{convex}}$
which results in~\eqref{eq:Ts}.
\end{proof}

With all behaviors now bounded, we can now prove Theorem \ref{theo:sufficient_horizon}.

\begin{proof}[Proof of Theorem \ref{theo:sufficient_horizon}] 
From Proposition~\ref{prop:T_D}, the system can decongest in $\overline{T}_D$ timesteps or fewer. From Proposition~\ref{prop:T_R}, the system can recover in $\overline{T}_R$ timesteps or fewer. From Proposition~\ref{prop:T_S}, the system needs at most $\overline{T}_S$ timesteps to have more cars exit the system over the entire rollout horizon than it would at the maximum possible outflow in the congested state over the same horizon. Thus, a rollout horizon of $\overline{T}_D + \overline{T}_R + \overline{T}_S$ is sufficient for a hysteretic controller to produce the optimal control action. 
\end{proof}
}{}






\subsection{Supply and Demand Generalizability}
The formulation of \(\overline{T}_D,\  \overline{T}_R,\ \text{and }\overline{T}_S\) in Section~\ref{section:horizon_req} assume the forms for the supply and demand functions outlined in \eqref{eq:supply} and \eqref{eq:demand}. These formulations, however, can be generalized to any valid supply and demand functions that fulfill the conditions for a valid CTM. In particular, only the demand function is used in the formulation. Thus, so long as the demand functions \(d_i(x)\) are monotonically increasing and are such that \(d_i(0) = 0\), the horizon calculations can be generalized. The formulation for \(\overline{T}_D\) takes the form
\begin{equation}
\label{eq:td_general}
    \overline{T}_D = \left \lceil{\frac{\beta_1x_c + (x_c - \underline{x}_2)}{hd_2(\underline{x}_2)-\overline{\lambda}_0}}\right \rceil \,,
\end{equation}
and  \(\overline{T}_R\) is now the smallest integer value that satisfies
\begin{equation}
\label{eq:tr_general}
    \begin{split}
    \overline{x}_2 - \underline{x}_2 \leq\ &-hd_2(\underline{x}_2) +\sum_{i=1}^{\overline{T}_R} [h(\beta_1d_1(x_1(\overline{T}_D+i)) \\
    \hspace{-1cm} &- d_2(x_2(\overline{T}_D+i)))] 
\end{split}
\end{equation}
where
\begin{equation}
    x_1(k) = x_1(k-1) 
    + c_1 + \lambda_0(k-1) -hd_1(x_1(k-1))\,.
\end{equation}

\(\overline{T}_S\) can be calculated using the relevant outflows at the target densities. Care should be taken when calculating the average outflows in each of the other periods if the supply and demand functions are in non-affine forms.

\medskip

Thus, we have provided theoretical conditions which, if fulfilled, result in a hysteretic controller outperforming one that ignores capacity drop. We have also provided a method for determining a sufficient rollout horizon such that the hysteretic controller is able to produce the optimal output. These theoretical results provide explicit horizon computations for a two-cell system, which we focus on in order to illustrate the effects of accounting for capacity drop locally. Nonetheless, the fundamental insight of Theorem \ref{theo:sufficient_horizon} that the EHMPC will outperform the RAMPC given a sufficiently large horizon applies to general networks, as demonstrated in the following case studies.


\section{Case Studies}
\label{section:case studies}
We have proven that there is significant outflow to be gained under certain conditions for a two-cell system. In this section, we provide empirical results for each controller on systems with more than two cells and illustrate the differences in behavior of each controller. The full set of simulation parameters for each system is found in Table~\ref{table:cell_params}.
\begin{figure}
    \begin{center}
        \input{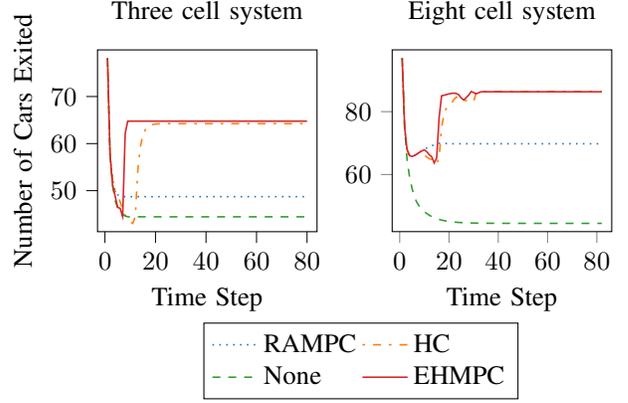}
    \end{center}
    \caption{Controller performance on each system.}
    \label{fig:sim_results}
    \iftoggle{extended}{}{\vspace{-0.5cm}}
\end{figure}

Fig.~\ref{fig:sim_results} displays the results of running each controller on a three-cell system and an eight-cell system, respectively, in order to illustrate performance gains between controllers. These systems begin in a congested state and remain congested, resulting in a suboptimal throughput rate without a controller. For the three-cell system the RAMPC was initialized with a rollout horizon of 51, and the EHMPC was initialized with a rollout horizon of 21 and a control memory (i.e. the number of control actions to execute before recalculating) of 5 actions. For the eight-cell system the EHMPC was initialized with a rollout horizon of 11 and no control memory, making it a true model predictive controller, and the RAMPC was initialized with the same rollout horizon as the three-cell system.

\begin{table}
 \caption{Simulation parameters used for each system.}
\label{table:cell_params}
\footnotesize
 \begin{tabular}{lll}
 \hline
 Parameter & Three-Cell System & Eight-Cell System\\ [0.5ex] 
 \hline
 sample time \(h\) & 1/120 hrs & 1/120 hrs\\ 
 number of timesteps \(K\) & 81 & 83\\
 number of cells $N$ & 3 & 8\\
 cells with onramps & cells 1, 2 & cells 1, 3, 5, 7\\
 onramp max flow \(c_i\) & 60 & 60\\
 onramp inflow $\lambda_{1...N}$ & \(80 \forall k\in [0, K)\) &  \(80 \forall k\in [0, K)\)\\
 upstream inflow $\lambda_0$ & 40 & 0\\
 \(\overline{x}_i\) & 110 & 110\\
 \(\underline{x}_i\) & 70 & 70\\
 free flow speed \(v_i\) & 60 & 60\\
 shockwave speed \(w_i\) & 20 & 20\\ 
 jam density \(x_i^{\text{jam}}\) & 320 & 320\\ 
 beta terms & \(\beta_1,\beta_2 = 0.9\) & \(\beta_1,...,\beta_7 = 0.9\)\\
 & \(\beta_3 = 1.0\) & \(\beta_8 = 1.0\)\\
 initial density \(x_i(0)\) & \(x_1(0) = 0, \) & 150 \\
 & \( x_2(0),x_3(0) = 150\) & \\ [0.5ex] 
 \hline
 \end{tabular}
 \centering
 \vspace{0.2cm}
 
 All parameters are consistent across every cell unless otherwise noted.
\end{table}

The Relaxed Approximate controller performs better than no controller in both systems by somewhat restricting the inflow of cars into the cells. Without accounting for hysteresis, however, this controller keeps the systems in congested states. This is because the RAMPC formulation does not differentiate between the congested and uncongested states and therefore does not recognize the benefit in driving the system to the uncongested state.
The Heuristic Controller decongests the systems in both cases, though it takes longer than needed. 

The Exact Hysteretic MPC recognizes the need for decongestion and severely limits the inflow from each onramp in order to push each system into its decongested state as quickly as possible, before returning it to steady-state. This steady-state is also at a much higher output when compared to the other control methods, as the controller's awareness of the system's hysteretic behavior allows it to take advantage of the greater free flow that can be found at higher densities. Even the Heuristic Controller's suboptimal steady-state flow still outperforms the Relaxed Approximate controller. Thus, we show that accounting for hysteresis, even imperfectly, produces a nonneglible improvement on the outflow of the system.

 
While the more complex eight-cell system produces interesting behaviors during decongestion and recovery, the EHMPC and the Heuristic Controller eventually achieve a higher steady-state output than the RAMPC. However, it is worth noting the increase in overall computation times between each controller, recorded in Table~\ref{table:completion_times}.
\begin{table}
\centering
 \caption{Time to Simulation Completion (seconds)}
 \begin{tabular}{ccc} 
 \hline
 Controller & Three-Cell & Eight-Cell \\ 
 \hline
 RAMPC & 76 & 85\\ 
 EHMPC & 445 & 6190\\
 HC & $<1$ & \(<1\)\\
 \hline
 \end{tabular}
 \label{table:completion_times}
\vspace{-0.6cm}
\end{table}
As the calculations performed in the EHMPC are more complex, the computation time increases between the two systems, whereas the computation time for the RAMPC increases only slightly. However, the Heuristic Controller outperforms the RAMPC in both outflow and computation time, thus illustrating that there is still a significant advantage to accounting for capacity drop, even if done so heuristically.

\section{Conclusions and Future Work}
\label{section:conclusions}
In this paper we have shown theoretically and with case studies 
that explicitly incorporating the capacity drop phenomenon when designing ramp metering strategies, particularly those based on model predictive control, leads to improved system performance. 
In particular, optimal control is achieved with a model predictive controller formulated as a mixed integer program with sufficiently long rollout horizon. Future directions of research include obtaining optimal or near-optimal controllers that account for capacity drop but avoid the computational complexity inherent in a large-scale mixed integer formulation. One approach might be to obtain distributed implementations similar to those that have been proposed in~\cite{rosdahl2018distributed, ba2016distributed}.

\bibliographystyle{ieeetr}
\vspace{-0.2cm}
\bibliography{ref.bib}

\iftoggle{extended}{
\appendix
\subsection{EHMPC Gurobi Implementation}
While Gurobi variables can be defined as binary, the solver does not allow complex boolean logic calculations or assigning comparisons such as \(\geq,\leq\) to variables as constraints. Thus, in order to accommodate the congestion state calculations we must define some additional variables and intermediate calculations.

We refer back to the congestion state calculation \eqref{eq:congestion} and see that it is essentially made up of three conditional statements. These conditionals must be implemented separately into the Gurobi solver, as the update equation is too complex to implement as a single constraint.

To implement the first conditional \(x_i(k)\geq\overline{x}_i\), we define \(\delta_c\) to be a congestion ``buffer" term and \(\overline{x}_{gi}\) to be a ``goal" term such that $\overline{x}_{gi} = \overline{x}_{i} - \delta_c$.

We choose $\delta_c$ to be a small enough value such that $\overline{x}_{gi}$ is close to $\overline{x}_{i}$, the importance of which is shown below. We then define the next few steps as follows:
\begin{equation*}
 \begin{split}
    a_{\overline{x}_i}(k) = (x_i(k) &- \overline{x}_{gi})/\delta_c\,, \quad 
    \overline{a}_i(k) = \min[a_{\overline{x}_i}(k), 1] \,, \\
    a_i(k) &= \max[\overline{a}_i(k), 0] \,.
     \end{split}
\end{equation*}
Separating each step in the above manner is crucial as each individual step can be input to Gurobi as a separate constraint. This operation results in an indicator variable \(a_i(k)\) such that
\begin{equation}
    \label{eq:above_x_upper_explanation}
    a_i(k) = 
        \begin{cases}
            1 & \text{if } x_i(k) \geq \overline{x}_i\,,\\
            0 & \text{if } x_i(k) \leq \overline{x}_{gi}\,,\\
            (x_i(k) - \overline{x}_{gi})/\delta_c & \text{otherwise}.
        \end{cases}
\end{equation}
Thus, so long as a sufficiently small \(\delta_c\) is chosen, \(a_i(k)\) is practically equivalent to the first conditional in Equation~\ref{eq:congestion}. We can similarly perform steps to model the second conditional \(x_i(k)\leq\underline{x}_i\) using the same \(\delta_c\):
\begin{equation*}
    \underline{x}_{gi} = \underline{x}_{i} \,, \quad
    b_{\underline{x}_i}(k) = (\underline{x}_{gi} - x_i(k))/\delta_c \,, 
\end{equation*}
\begin{equation*}
    \overline{b}_i(k) = \min[b_{\underline{x}_i}(k), 1] \,, \quad
    b_i(k) = \max[\overline{b}_i(k), 0] \,,
\end{equation*}
which results in an indicator variable \(b_i(k)\) such that
\begin{equation}
    \label{eq:below_x_lower_explanation}
    b_i(k) = 
        \begin{cases}
            1 & \text{if } x_i(k) \leq \underline{x}_i - \delta_c \,, \\
            0 & \text{if } x_i(k) \geq \underline{x}_{gi} \,,\\
            (\underline{x}_{gi} - x_i(k))/\delta_c & \text{otherwise}.
        \end{cases}
\end{equation}

It should be noted that \(a_i(k)\) is flagged as true (i.e. equal to 1) when the density reaches \(\overline{x}\) exactly, whereas there is a bit of a buffer below \(\underline{x}\) before \(b_i(k)\) is flagged as true. This is mainly done for safety reasons in the controller, in case of floating point error between the controller and the simulator.

Finally, we define a binary variable $z_i(k) = a_i(k) \lor b_i(k)$.

Thus, having defined \(a_i(k), b_i(k), z_i(k)\), Equation~\eqref{eq:congestion} can be implemented as a set of General Indicator Constraints.

Similarly, the outflow calculation must be separated into individual components in order to be implemented. To do so, we define two variables, representing the decongested and congested outflow, respectively, of each cell at each timestep such that $\phi^d_i(k) = d_i(k)$ and $\phi^c_i(k) = \min(d_i(x_i(k)),s_{i+1}(x_{i+1}(k))/\beta_i)$.
The true outflow calculation can then be set as a set of General Indicator Constraints utilizing \(\sigma_i(k)\).

Thus, the system behavior is entirely encoded into the Gurobi framework. This allows for the optimal control action to be calculated for the system, provided that the MPC rollout horizon is sufficiently long, as discussed in Section~\ref{section:theoretical_results}. However, because the above formulation is a nonconvex mixed integer linear program, the computation time is generally longer for the Exact Hysteretic MPC than the convex Relaxed Approximate MPC.
}{}

\end{document}